\documentclass[AMA,STIX1COL]{WileyNJD-v2}

\usepackage{moreverb,url}

\usepackage{mathrsfs,bm}
\usepackage{dsfont}
\usepackage{float}
\usepackage{bbm}
\usepackage[latin1]{inputenc}
\usepackage{amsmath}
\usepackage{graphicx,subfig}
\usepackage{amsthm}
\usepackage{verbatim}
\usepackage{epsfig,epstopdf}
\usepackage[labelfont=bf]{caption}
\usepackage{enumerate}
\usepackage{epstopdf}
\usepackage{tikz,circuitikz}
\makeatletter
\newcommand{\pushright}[1]{\ifmeasuring@#1\else\omit\hfill$\displaystyle#1$\fi\ignorespaces}
\makeatother

\def\begquo{\begin{quote}}
	\def\endquo{\end{quote}}
\def\begequarr{\begin{eqnarray}}
	\def\endequarr{\end{eqnarray}}
\def\begequarrs{\begin{eqnarray*}}
	\def\endequarrs{\end{eqnarray*}}
\def\begarr{\begin{array}}
	\def\endarr{\end{array}}
\def\begequ{\begin{equation}}
	\def\endequ{\end{equation}}
\def\lab{\label}
\def\begdes{\begin{description}}
	\def\enddes{\end{description}}
\def\begenu{\begin{enumerate}}
	\def\begite{\begin{itemize}}
		\def\endite{\end{itemize}}
	\def\endenu{\end{enumerate}}

\def\lef[{\left[\begin{array}}
	\def\rig]{\end{array}\right]}

\def\begcen{\begin{center}}
	\def\endcen{\end{center}}
\def\begrem{\begin{remark}\rm}
	\def\endrem{\end{remark}}
\def\begdef{\begin{definition}}
	\def\enddef{\end{definition}}
\def\begpro{\begin{proposition}}
	\def\endpro{\end{proposition}}
\def\begfac{\begin{fact}}
	\def\endfac{\end{fact}}
\def\begass{\begin{assumption}}
	\def\endass{\end{assumption}}

\def\begsubequ{\begin{subequations}}
	\def\endsubequ{\end{subequations}}

\def\begmat#1{\begin{bmatrix}#1\end{bmatrix}}
\def\begali#1{\begin{align}{#1}\end{align}}
\def\begalis#1{\begin{align*}{#1}\end{align*}}

\def\calg{{\cal G}}

\def\calc{{\cal C}}

\def\calr{{\cal R}}
\def\calj{{\cal J}}

\def\calw{{\cal W}}
\def\call{{\cal L}}
\def\calf{{\cal F}}

\def\calj{{\cal J}}

\def\caly{{\cal Y}}

\def\tilthe{\tilde{\theta}}

\def\dottilthe{\dot{\tilde{\theta}}}
\def\liminf{\lim_{t \to \infty}}

\def\L2e{{\cal L}_{2e}}

\def\rea{\mathds{R}}

\def\adj{\mbox{adj}}
\def\col{\mbox{col}}
\def\hal{{1 \over 2}}

\def\TAC{{\it IEEE Trans. Automatic Control}}

\def\IJC{{\it International Journal of Control}}

\def\SCL{{\it Systems and Control Letters}}
\def\AUT{{\it Automatica}}

\def\calg{{\cal G}}

\def\hal{{1 \over 2}}

\usepackage{color}

\usepackage[prependcaption,colorinlistoftodos]{todonotes}

\def\calp{{\mathfrak p}}
\articletype{Article Type}%

\received{26 april 2016}
\revised{6 june 2016}
\accepted{6 june 2016}
\begin{document}

\title{Indirect Adaptive Control of Nonlinearly Parameterized Nonlinear Dissipative Systems.}

\author[1]{Romeo Ortega}

\author[1]{Rafael Cisneros*}

\author[2]{Lei Wang}

\author[3]{Arjan van der Schaft}

\authormark{R. Ortega \textsc{et al}}

\address[1]{\orgdiv{Departamento Académico de Sistemas Digitales}, \orgname{ITAM}, \orgaddress{\state{Ciudad de M\'exico}, \country{M\'exico}}}

\address[2]{\orgdiv{Australian Center for Field Robotics}, \orgname{The University of Sydney}, \orgaddress{\state{Sydney}, \country{Australia}}}

\address[3]{\orgdiv{Bernoulli Institute for Mathematics, Computer Science and AI}, \orgname{University of Groningen}, \orgaddress{\state{Groningen}, \country{The Netherlands}}}

\corres{*R. Cisneros, Departamento Académico de Sistemas Digitales, ITAM, Río Hondo 1, 01080 Ciudad de M\'exico, M\'exico. \email{rcisneros@itam.mx}}


\abstract[Summary]{In this note we address the problem of indirect adaptive (regulation or tracking) control of nonlinear, input affine {\textit {dissipative}} systems. It is assumed that the supply rate, the storage and the internal dissipation functions may be expressed as \textit{ nonlinearly parameterized} regression equations where the mappings (depending on the unknown parameters) satisfy a \textit{ monotonicity} condition---this encompasses a large class of physical systems, including passive systems. We propose to estimate the system parameters using the ``power-balance'' equation, which is the differential version of the classical dissipation inequality, with a new estimator that ensures global, exponential, parameter  convergence under the very weak assumption of \textit{interval excitation} of  the power-balance equation regressor.   To design the indirect adaptive controller we make the standard assumption of existence of an asymptotically stabilizing controller that depends---possibly nonlinearly---on the unknown plant parameters, and apply a certainty-equivalent control law.  The benefits of the proposed approach, with respect to other existing solutions, are illustrated with examples.}


\maketitle

%
\section{Introduction}
\lab{sec1}
%
The problem of adaptive control of nonlinear systems has attracted the attention of researchers for several years, see \cite{ASTKARORT,KRSKANKOKbook,MARTOMbook,PRAetal} for a survey of the literature.  The development of \textit{ direct} adaptive controllers, where we estimate directly the parameters of a full information stabilizing controllers, is stymied by the so-called matching condition \cite[Section 3.3]{ASTKARORT}, which imposes severe constraints on the class of systems for which it is applicable. This obstacle is avoided for the case of systems with particular (triangular) structures \citep{KRSKANKOKbook}, but this structural assumption is just mathematically motivated and is rarely verified in physical systems. Although it is possible, in some cases,  to transform a general nonlinear system into a triangular one this requires the solution of a partial differential equations, which is difficult to find. For this reason, must of the recent attention has been centered on \textit{indirect} adaptive controllers, where we estimate the parameters of the plant and then compute the parameters of the controller.  

The implementation of indirect adaptive  controllers is, in general, complex and very computationally demanding. This is mainly due to the fact that the parameterization that is used to obtain the linear regression equation (LRE) needed for their implementation is based on the state space model of the system dynamics,  which involves complicated signal and parameter relations that are unrelated with the physical properties of the system. An additional difficulty is that, in order to obtain a \textit{linear} relation in the LRE, it is often necessary to \textit{overparametrize} the vector of unknown parameters. This approach has very serious shortcomings, in particular the need of more \textit{stringent excitation conditions} stemming from the fact that the parameter search takes place in a bigger dimensional space with nonunique minimizing solutions---see  \cite{LJUbook,SASBODbook} and the detailed discussion in  \cite[Section 1]{ORTetal}. This situation has severely stymied the practical implementation of  adaptive control techniques in many critical applications. 

A route, pursued by some control researchers, to overcome this difficulty is to replace the complicated expression of the regressor using function approximation,  for instance neural networks or fuzzy controllers.  Unfortunately, as always with function approximation-based techniques \cite{ORT}, although they might lead to successful designs, there is no solid theoretical guarantee that  the procedure will work. A second alternative is to use high-gain based schemes, like sliding modes or fractional power controllers that---as is well-known \cite{ARAetalijc}---suffer from their extreme sensitivity to the unavoidable presence of noise in the system, rendering them unfeasible in must practical applications.

For robotic applications it was suggested in  \cite[Section 2.2]{SLOLIaut} to use the parameterization of the \textit{power-balance equation} to design an indirect adaptive controller.  In an independent line of research,  in \cite[Subsection 12.6.2]{KHADOMbook} this parameterization  was used for the \textit{identification} of the robot parameters. The main advantage of this approach is that the resulting parameterization avoids the cumbersome terms related to the Coriolis and centrifugal forces matrix. This is a significant simplification that drastically reduces the complexity and computational demands.  To the best of our knowledge, such an approach was never actually pursued, because the \textit{excitation requirements} for the consistent estimation of the parameters is ``very high"---see \cite[Remark 16]{ORTetal}. In the recent paper  \cite{ROMORTBOB} a procedure to overcome, for the first time, this fundamental problem was proposed. Towards this end, a recent technique of generation of new LRE with ``exciting" regressors developed in \cite{BOBetal} was used. 

The main contributions of the paper, which include several generalizations of the results in   \cite{ORTetal,ROMORTBOB}, may be summarized as follows. 
\begenu[{\bf C1}]
\item As indicated above most of the research on adaptive control has relied on the use of LREs, which are usually obtained overparameterizing the  regression equations. In contrast with this approach we consider here the case where the uncertain parameters enter into the system dynamics in a \textit{nonlinear} way and construct a \textit{ nonlinearly  parameterized regression equation} (NLPRE). The interested reader is referred to \cite{BOFSLO,ORTetal} for recent reviews of the literature dealing with NLPRE. 
\item In contrast to the results in   \cite{ORTetal,ROMORTBOB} where, to prove  \textit{parameter convergence}, it is necessary to assume some \textit{a priori non-verifiable}  conditions, in this paper we use the parameter estimator proposed in \cite{WANetal}---called G+D estimator---that ensures (global exponential) parameter convergence  assuming only the extremely weak condition of \textit{interval excitation} \cite{KRERIE,TAObook} of the original vector regressor. An additional advantage of the G+D estimator is that it can deal with a class of NLPREs---in particular, separable ones.
\item We extend the technique, restricted in   \cite{ROMORTBOB} to Euler-Lagrange systems, to the much broader class of \textit{dissipative} systems \cite{MOYbook,VANbook}, which contains as a particular case Euler-Lagrange and port-Hamiltonian systems.
\item The use of the power balance equation, instead of the full dynamics of the system, to obtain the parameterization needed for an indirect adaptive control implementation, yields a significant \textit{numerical complexity reduction}. The impact of such a simplification in the practical feasibility and the excitation requirements of the scheme can hardly be overestimated.  
\endenu

The remainder of the paper is organized as follows. In Section \ref{sec2} we identify the class of systems for which our adpative control result is applicable and present the problem formulation.  In  Section \ref{sec3} we present the derivation of a NLPRE for the estimation of the unknown plant parameters proceeding from the   dissipation inequality. We also recall in this section  the standard parameterization  that imposes the, rarely verified,  assumption of linearity in the parameters of the system dynamics.  Section \ref{sec4} presents the proposed adaptive control scheme. Simulation results, which illustrate the {performance} of the proposed controller  are presented in Section \ref{sec5}. The paper is wrapped-up with concluding remarks in  Section \ref{sec6}.\\

\noindent {\bf Notation.} $I_n$ is the $n \times n$ identity matrix.  For $x \in \rea^n$, we denote the square of the Euclidean norm as $|x|^2:=x^\top x$. The action of an operator $\calf:\call_{\infty e} \to \call_{\infty e}$ on a signal $u(t)$ is denoted as $\calf[u]$. In particular, we define the derivative operator $\calp[u]=:{du(t)\over dt}$. All mappings are assumed smooth and all dynamical systems are assumed to be forward complete. Given a function $f:  \rea^n \to \rea$ we define the differential operator $\nabla f:=\left(\frac{\displaystyle \partial f }{\displaystyle \partial x}\right)^\top$.  To simplify the notation, whenever clear from the context, the arguments of the various functions are omitted.
%
\section{Problem Formulation}
\lab{sec2}
%
\subsection{Plant description}
\lab{subsec21}
Consider the input affine nonlinear system
\begin{align}
	\nonumber
	\dot x &= f(x,\theta)+g(x,\theta)u_p\\
	y_p & = h(x,\theta)+j(x,\theta)u_p,
	\lab{sys}
\end{align}
where $x(t)\in\mathbb{R}^{n}$ is the systems state $u_p(t)\in\mathbb{R}^{n_p}$ and $y_p(t)\in\mathbb{R}^{n_p}$ are the port variables,  $\theta\in\mathbb{R}^q$ is a vector  of \textit{unknown}  parameters and 
\begalis{
	& f:\rea^n \times \rea^q \to \rea^n,\; g:\rea^n \times \rea^q \to \rea^{n \times n_p} \\
	& h:\rea^n \times \rea^q \to \rea^{n_p},\, j:\rea^n \times \rea^q \to \rea^{n_p\times n_p}.
}
To be able to treat systems with external sources, the vector $u_p$ is assumed to be of the form $u_p=\col(u,E)$, where $u(t)\in\mathbb{R}^{m}$ are the control signals and $E(t)\in\mathbb{R}^{n_p-m}$ are external signals that represent uncontrollable sources (or loads). \\
\subsection{Assumptions}
\lab{subsec22}
We make the following assumptions on the system.

\begenu
\item[{\bf A1}]  [Measurements] The systems state $x$ is \textit{measurable}.

\item[{\bf A2}]  [Dissipativity]  The system \eqref{sys} is \textit{dissipative}  \cite[Definition 3.1.2]{VANbook} with respect to the supply rate $s:\mathbb{R}^{n_p}\times\mathbb{R}^{n_p}\times \mathbb{R}^{q} \to\mathbb{R}$, \textit{i.e.}, there exist a storage function ${S}:\mathbb{R}^{n}\times \mathbb{R}^{q}\to\mathbb{R}_+$  and  an internal dissipation function $d:\mathbb{R}^{n}\times \mathbb{R}^{q}\to\mathbb{R}_+$ such that\
\begequ
\lab{powbal}
\dot S(x,\theta)= -d(x,\theta)+ s(u_p,y_p,\theta).
\endequ
\item[{\bf A3}]  [Parameterization] The functions $s$,  ${S}$ and $d$ admit the following \textit{separable} NLPRE representation 
\begsubequ
\lab{linpar}
\begin{align}
	\label{linpars}
	{s}(u_p,y_p,\theta)&=\phi_s^\top(u_p,y_p)\calg_s(\theta) + b_s(u_p,y_p)\\
	\label{linparS}
	{S}(x,\theta)&=\phi_S^\top(x)\calg_S(\theta) + b_S(x)\\
	\label{linpard}
	d(x,\theta)&=\phi^\top_d(x)\calg_d(\theta) + b_d(x),
\end{align}
\endsubequ
where the functions 
\begalis{
	& \phi_s:\mathbb{R}^{n_p}\times\mathbb{R}^{n_p} \to \rea^{p_s},\;\phi_S:\rea^n \to \rea^{p_S},\;\phi_d:\rea^n \to \rea^{p_d}\\
	& b_s:\mathbb{R}^{n_p}\times\mathbb{R}^{n_p} \to \rea,\;b_S:\rea^n \to \rea,\;b_d:\rea^n \to \rea
} 
and the mappings $\calg_s:\rea^q \to \rea^{p_s}$,  $\calg_S:\rea^q \to \rea^{p_S}$ and  $\calg_d:\rea^q \to \rea^{p_d}$ are \textit{known} and 
$$
p:=p_s+p_S+p_d \geq q.
$$ 
\item[{\bf A4}]  [Monotonicity] Define the mappings $\calg:\rea^q \to \rea^{p} $
\begequ
\lab{calg}
\calg(\theta):=\begmat{\calg_s(\theta) \\  \calg_S(\theta) \\ \calg_d(\theta)}
\endequ
and $\calw:\rea^q \to \rea^{q}$
$$
\calw(\theta):=T\calg(\theta),
$$
where $T\in\mathbb{R}^{q\times p}$ is \textit{chosen by the designer}. Assume there is a positive definite matrix $P\in\mathbb{R}^{q\times q}$ such that $\calw(\theta)$ is \textit{strongly $P$-monotonic} \cite{DEM,PAVetal}. That is,
\begali{
	\lab{monpro-S}
	(a-b)^\top P &\left[\calw(a) - \calw(b)\right] \geq \rho|a-b|^2 >  0,\;\forall \; a,b \in \rea^q,
}
with $a \neq b$ and for some $\rho {>0}$.
\item[{\bf A5}]  [Stabilizability] Given a desired bounded trajectory for the state vector $x_\star(t) \in \rea^{n}$, with bounded derivative. Define the state tracking error $\tilde x:=x-x_\star.$ There exists a mapping $\beta: \rea^{n} \times \rea^q \times \rea_{\ge 0} \to \rea^m$, such that the closed-loop system
$$
\dot x= f(x,\theta)+g(x,\theta)\begmat{ \beta(x,\theta,t) \\ E}
$$
has an \textit{error dynamics} 
\begequ
\lab{errsys}
\dot {\tilde x}=f_\star(\tilde x,\theta,E,t),
\endequ
whose origin is uniformly asymptotically stable (UAS). 
\endenu
\subsection{Control objective}
\lab{subsec23}
Design an  estimator of the plant parameters $\theta$ 
\begali{
	\nonumber
	\dot {x}_\theta & =f_\theta(x_\theta,\tilde x,t)\\
	\lab{paa}
	\hat \theta & =h_\theta(x_\theta,\tilde x,t),
}
where $f_\theta: \rea^{n_\theta} \times \rea^n \times \rea_{\ge 0} \to \rea^{n_\theta}$ and $h_\theta:  \rea^{n_\theta} \times\rea^{n}  \times \rea_{\ge 0} \to \rea^q$ which ensures \textit{global, exponential convergence} of the parameter errors $\tilde \theta:=\hat \theta-\theta$ and such that the (certainty-equivalent) indirect adaptive control $u=\beta(x,\hat \theta,t)$ ensures UAS of the zero equilibrium of the adaptively controlled error system
\begali{
	\dot {\tilde x} & =f_\star(\tilde x,\theta,E,t)+g(\tilde x+x_\star,\theta)\begmat{ \beta(\tilde x+x_\star,\tilde \theta+\theta,t)  -\beta(\tilde x+x_\star,\theta,t)\\ 0}.
	\lab{errsysada}
}
Consequently,  
\begequ
\lab{asycon}
\liminf \tilde x(t) = 0,
\endequ
with all signals bounded provided the initial errors $\tilde x(0)$ are sufficiently small. 
\subsection{Discussion}
\lab{subsec24}
The following remarks are in order.\\

\noindent {\bf P1}
We have restricted ourselves to  a \textit{local} stabilization objective using static---versus dynamic---state-feedback controllers. As will become clear below, the extension to the case where the controller is \textit{dynamic}  follows \textit{verbatim}. However, to obtain \textit{global} results in tracking it is necessary to strengthen the UAS Assumption {\bf A5} to global \textit{exponential} stability.  This additional assumption is not necessary in \textit{regulation} tasks with static state-feedback controllers for which  global \textit{asymptotic} stabilization of the known parameter controller is sufficient for global stabilization of the adaptive one.  See \cite[Remark 12]{ORTetal} for a discussion on this issue.\\

\noindent {\bf P2}
A convenient representation of the supply rate $s$ is the one corresponding to the so-called \textit{(QSR) dissipativity} \cite{MOYbook} that is given as
$$
s(u_p,y_p)=y^\top _pQy_p+2y^\top _pSu_p+u_p^\top Ru_p,
$$ 
where $Q, S$ and $R$ are $n_p \times n_p$ constant matrices, with $Q$ and  $R$ symmetric. It is clear that the properties of passivity and finite $\call_2$-gain are particular cases of (Q,S,R) dissipativity \cite{MOYbook,VANbook}. In reference to the first property, and with some abuse of notation, we will refer to \eqref{powbal} as \textit{``power-balance'' equation}, which is the differential version of the classical dissipation inequality including the dissipation function $d$. \\

\noindent {\bf P3}
As is well-known \cite{ASTKARORT,SASBODbook} the difference between direct and indirect adaptive controllers is that, in the former, it is assumed that there exists a parameter-dependent controller that achieves the control objective, while in the latter we additionally assume that the plant depends on some unknown parameters and that there exists a mapping between the plant and the controller parameters that allows us to compute the control signal.  In direct adaptive control we estimate directly the parameters of the controller. On the other hand, in its indirect version the parameters of the plant are estimated and then the parameters of the controller are computed via the aforementioned mapping.  To simplify the presentation in the problem formulation above we have obviated this latter step---embedding the mapping between the plant and controller parameters in the function $\beta$.\\

\noindent {\bf P4}
The Assumption {\bf A4} of  strong $P$-monotonicity of  $\calw(\theta)$ is similar to the one used in \cite{ORTetal},  The form adopted here was first proposed in \cite{WANetal}. As indicated in  \cite{DEM,ORTetal,PAVetal} it is possible to verify the monotonicity assumption from the Jacobian of $\calw(\theta)$ invoking  \cite[Lemma 1]{ORTetal}.

\section{Derivation of the New Systems Parameterization}
\lab{sec3}
%
In this section we present the first contribution of the paper, namely the derivation of a NLPRE for the estimation of the unknown parameters $\theta$ proceeding from the   ``power-balance'' equation \eqref{powbal}---to which we will refer in the sequel as \textit{ power-balance equation parameterization (PBEP)}. As indicated in the Introduction this is a generalization of the procedure proposed for robot manipulators in \cite{KHADOMbook, SLOLIaut}, see also  \cite{ROMORTBOB}. We also recall in Subsection  \ref{subsec32} that, imposing the assumption that the functions $f$ and $g$ of \eqref{sys} admit a linear parameterization---usually doing some overparameterization---it is possible to apply standard filtering techniques to derive a LRE \textit{ without} the requirement of dissipativity. The advantages of the new PBEP, with respect to the latter one, are discussed in Subsection \ref{subsec33} below. 
\subsection{Power-balance equation parameterization}
\lab{subsec31}
%
\begin{proposition}
	\label{pro1}
	Consider the nonlinear system \eqref{sys} satisfying Assumptions {\bf A1}-{\bf A3}. Fix an LTI, stable filter
	\begequ
	\lab{f}
	F(\calp)={\lambda \over \calp+\lambda},
	\endequ
	with $\lambda>0$. Define the signals
	\begsubequ
	\lab{dynext}
	\begali{
		\lab{dynexty}
		Y &:=F(\calp)[b_s - b_d]- \calp F(\calp)[ b_S] \\
		\lab{dynextome}
		\Omega &:=\begmat{ -F(\calp)[\phi _s] \\  \calp F(\calp)[\phi _S]\\ F(\calp)[\phi_d] }.
	}
	\endsubequ
	The following NLPRE holds
	\begequ
	\lab{lre}
	Y =\Omega^\top   \calg(\theta),
	\endequ
	where $\calg(\theta)$ is defined in \eqref{calg}.
\end{proposition}

\begin{proof}
	We carry out the next operations
	\begalis{
		\dot S&= -d+ s  \qquad \qquad\qquad\qquad \qquad\qquad \qquad\;\; \qquad\quad\;\;(\Leftrightarrow\;  \eqref{powbal})\\
		\calp F(\calp)[ S] & = -F(\calp)[ d]+F(\calp)[s] \qquad \qquad \qquad\qquad\qquad \qquad\;\;\;(\Leftarrow\;F(\calp)[\cdot])\\
		\calp F(\calp)[\phi_S^\top\calg_S(\theta) + b_S] & = -F(\calp)[\phi^\top_d\calg_d(\theta)  + b_d]+ F(\calp)[\phi^\top_s\calg_s(\theta)  + b_s] \quad \qquad(\Leftarrow\; \eqref{linpar})\\
		F(\calp)[b_s] - F(\calp)[ b_d]- \calp F(\calp)[ b_S]&=\begmat{ -F(\calp)[\phi^\top _s] & \calp F(\calp)[\phi^\top _S]& F(\calp)[\phi^\top_d] } \calg(\theta) \quad \quad (\Leftarrow\;\theta=const,\;\eqref{calg}) \\
		Y &= \Omega^\top \calg(\theta)  \qquad \qquad\qquad\qquad \qquad \qquad \quad\qquad \qquad\;(\Leftrightarrow\;\eqref{dynext}, \eqref{lre}).
	}
\end{proof}
\subsection{Standard linear parametrization}
\lab{subsec32}
%
To derive the standard parameterization we make the following assumption.\\

\noindent {\bf A6} [Standard LPRE] The vector field  ${f}$ and the elements of the matrix $g$, that is, $g_{ij}$ for $i=1,\dots,n,\;j=1,\dots,n_p$, admit the following  parameterization 
\begsubequ
\lab{linparsys}
\begin{align}
	\label{linparf}
	{f}(x,\theta)&={\bm w _{f}}(x)\calc(\theta) +{\bf  b}_{f}(x)\\
	\label{linparg}
	g_{ij}(x,\theta)&=\phi^\top_{g_{ij}}(x)\calc( \theta) + b_{g_{ij}}(x),
\end{align}
\endsubequ
with the mapping $\calc: \rea^{q} \to \rea^{n_w},\;n_w \geq q$, and \textit{ known} functions 
\begalis{
	&{\bm w_{f}}:\rea^{n}\to \rea^{n \times n_w},\;\phi_{g_{ij} }:\rea^{n} \to \rea^{n_w}\\
	& {\bf  b}_{f} :\rea^n \to \rea^n,\;b_{g_{ij} }:\rea^n \to \rea.
}  

To streamline the presentation of the result we define the matrices
\begequ
\lab{inpmat}
{\bm w_g }(x,u_p):=\begmat{ \sum_{j=1}^{n_p}\phi^\top _{g_{1j}}(x)u_{pj}  \\ \vdots  \\  \sum_{j=1}^{n_p}\phi^\top _{g_{nj}}(x)u_{pj}},\;
{\bf B}_g(x):=\begmat{ b_{g_{11}}(x) & \cdots & b_{g_{1n_p} }(x) \\ \vdots & \vdots & \vdots \\ 
	b_{g_{n1}}(x) & \cdots & b_{g_{nn_p} }(x)},
\endequ
and notice that ${\bm w_g }:\rea^{n} \times\rea^{n_p}\to \rea^{n \times n_w}$ and $ {\bf B}_g :\rea^n \to \rea^{n \times n_p}$.

\begin{proposition}
	\label{pro2}
	Consider the nonlinear system \eqref{sys} verifying Assumption  {\bf A6}. Define the signals
	\begali{
		\nonumber
		{\bf Y}&:= \calp F(\calp)[x]- F(\calp)[{\bf  b}_{f}+ {\bf B}_g u_p] \\
		\lab{dynexts}
		{\bm \Omega}^\top  &:=F(\calp)[{\bm w_f }+{\bm w_g }]
	}
	where $F(\calp)$ is defined in \eqref{f}. Define the extended vector of parameters
	\begequ
	\lab{The}
	\Theta:=\calc(\theta).
	\endequ 
	The following LRE holds
	\begequ
	\lab{lres}
	{\bf Y} ={\bm \Omega}^\top  \Theta.
	\endequ
\end{proposition}

\begin{proof}
	The following set of operations is carried out 
	\begin{align*}
		\dot x  =& \left(\bm{w}_f +\bm{w }_g \right)\Theta+\bm{b}_f+{ \bf B}_gu_p \qquad \hspace{2.4cm}(\Leftrightarrow\;\eqref{sys},\eqref{linparsys},\eqref{The})\\
		\calp F(\calp)[ x]  =&F(\calp)[ \left(\bm{w }_f +\bm{w }_g \right)\Theta]+F(\calp)[\bm{b}_f+{ \bf B}_gu_p] \qquad \qquad (\Leftarrow\;F(\calp)[\cdot])\\
		\calp F(\calp)[x]-F(\calp)[\bm{b}_f+ { \bf B}_gu_p]  =& F(\calp)[\bm{w }_f +\bm{w }_g ]\Theta \hspace{4.5cm}(\Leftarrow\;\Theta=const) \\
		{\bf Y} =&  \bf{\Omega}^\top \Theta \hspace{6.4cm} (\Leftrightarrow\;\eqref{dynexts}).
	\end{align*}
\end{proof}

\subsection{Discussion}
\lab{subsec33}
%
The following remarks are in order.\\

\noindent {\bf P5} The main advantages of the PBEP with respect to the standard one include the following:
\begite
\item[(i)] The \textit{reduced complexity} in the derivation of the corresponding parameterization \eqref{lre} and \eqref{lres} is evident comparing \eqref{dynext} with \eqref{dynexts}. The examples given in Section \ref{sec5} will further illustrate this point.

\item[(ii)] Mathematical modeling of physical systems usually proceeds from a classification of its components into energy-storing and energy-dissipating, which appear explicitly in the storage and dissipation functions, respectively. These elements are interconnected among themselves and  the external sources via the physical laws, \textit{ e.g.}, Kirchhoff's or Newton's---see \cite{VANJELbook}.  In many cases, including mechanical, electrical and electromechanical systems, the system parameters  verify the ``monotonicity" Assumption {\bf A4}. 

\item[(iii)] In contrast with the remark above, in the state space description \eqref{sys} the physical parameters will enter the functions $f$ and $g$ multiplied among themselves rendering harder the verification of  Assumption {\bf A6} . Very often, \textit{e.g.}, in robotics, it is possible to \textit{ overparameterize} the functions to comply with the linearity requirement.  It is well-known that overparameterization has very severe shortcomings, see \cite{LJUbook,ORTetal} for a detailed discussion on this point.
\endite

\noindent {\bf P6} A state-space realization of  \eqref{dynexty} is given by 
\begalis{
	\dot x_Y &= - \lambda (x_Y+b_S) + b_d - b_s\\
	Y &=-{\lambda}(x_Y+b_S).
}

On the other hand, a state-space realization for   \eqref{dynextome} is
	\begin{align*}
		\dot x_{\Omega} &= \begin{bmatrix}-\lambda x_{\Omega_1}-\phi_s\\-\lambda (x_{\Omega_2}-\phi_S)\\-\lambda x_{\Omega_3}+\phi_d\end{bmatrix}\\
		\Omega &=\lambda\;\col(x_{\Omega_1},\phi_S-x_{\Omega_2},x_{\Omega_3}).
	\end{align*}
	It is clear that the state space description of  \eqref{dynexts} is very similar to the one given above  for \eqref{dynext}, hence it is omitted for brevity.\\
	
	\noindent {\bf P7} As thoroughly discussed in  \cite{VANbook,VANJELbook}, many physical systems can be described by port-Hamiltonian models of the form
	\begalis{
		\dot x &=[\calj(x) - \calr]\nabla H(x)+g(x)u_p\\
		y_p&=g^\top (x)\nabla H(x),
	}
	where $\calj(x)=-\calj^\top(x)$ and  $\calr=\calr^\top \geq 0$ is a constant matrix, representing the interconnection and dissipation structures and $H:\rea^n \to \rea_+$ is the energy function of the system. These systems are passive and they satisfy the power-balance equation \eqref{powbal} with 
	$$
	s=u_p^\top y_p,\;S=H,\;d=\nabla H^\top \calr \nabla H.
	$$
	In many practical examples the energy function is of the form $H(x)=\hal x^\top Q x$, with $Q>0$. In this case, Assumption  {\bf A4} is satisfied introducing a reparameterization of the form $Q^\top \calr Q$ for the elements of the dissipation function $d$.\\
	
	\noindent {\bf P8} It is clear from the derivations above that $F(\calp)$ can be replaced in both propositions by any strictly proper stable LTI filter. This degree of freedom can be exploited to attenuate the deleterious effect of measurement noise.\\
	
	\noindent {\bf P9} In some applications part of the external sources, denoted $E \in \rea^{n_p-m}$ in the system description of Section \ref{sec2}, are \textit{ unkown}. If they enter in the supply rate in the form
	$$
	s(u_p,y_p)=s_k(u,y_p)+s _u^\top(u,y_p)E,
	$$
	with known functions $s_k$ and $s _u$ and they are constant, it is possible to incorporate these uncertain parameters into the vector $\theta$ and derive a new NLPRE that includes them. \\
	
	\noindent {\bf P10}
	For the sake of clarity of presentation, in Assumption {\bf A3} we suppose that the supply rate $s$ and the dissipation functions $d$ admit \textit{ independent} parameterizations of the form \eqref{linpars} and \eqref{linpard}. From the proof of Proposition \ref{pro1} it is clear that we can replace this by the existence of a parameterization of the form 
	$$
	{s}(u_p,y_p,\theta)+ d(x,\theta)=\phi_{sd}^\top(u_p,y_p,x)\calg_{sd}(\theta) + b_{sd}(u_p,y_p,x),
	$$
	with known functions $\phi_{sd}$ and  $b_{sd}$. This variation will, in general, yield simpler expressions for the regressors. 
	%
	\section{Indirect Adaptive Control}
	\lab{sec4}
	%
	In this section we present the second main contribution of the paper, namely the construction of a UAS indirect adaptive controller for systems satisfying Assumptions {\bf A1}-{\bf A5} using the PBEP \eqref{lre} of Proposition \ref{pro1}.  As indicated in point {\bf C2} of the Introduction the key step is the utilization of the G+D estimator of \cite[Proposition 7]{WANetal} that we briefly recall in the lemma below---whose  proof is given in the previous reference. Replacing this estimates in the controller of  Assumption {\bf A5} yields the proposed adaptive controller.
	\subsection{The G+D estimator of \cite[Proposition 7]{WANetal}}
	\lab{subsec41}
	%
	As expected from an identification-based procedure some excitation assumptions will be required. However,    as shown in \cite{WANetal} it is possible to achieve global, exponential convergence of the parameter error imposing the following extremely weak \textit{ interval excitation} assumption \cite{KRERIE,TAObook} of the regressor vector $\Omega$ of the NLPRE \eqref{lre}.  \\
	
	\noindent {\bf A7} [Boundedness and Interval Excitation] The regressor vector $\Omega$ of the NLPRE \eqref{lre} is bounded\footnote{Boundedness of $\Omega$ is a blanket assumption made to avoid technicalities in the proofs.} an interval exciting. That is, there exists constants $C_c>0$ and $t_c>0$ such that
	\begalis{
		&\int_0^{t_c} \Omega(s) \Omega^\top(s)  ds \ge C_c I_q.
	}
	
	\begin{lemma}
		\lab{lem1}
		Consider the NLPRE (\ref{lre}) with  $\calg(\theta)$ satisfying Assumption {\bf A4} and $\Omega$ verifying Assumption  {\bf A7}. Define the G+D interlaced estimator
		\begsubequ
		\lab{intestt}
		\begali{
			\lab{thegt}
			\dot{\hat \theta}_g & =\gamma_g \Omega (Y-\Omega^\top \hat\theta_g),\; \hat\theta_g(0)=\theta_{g0} \in \rea^p\\
			\lab{phit}
			\dot { \Phi}& =  -\gamma_g \Omega \Omega^\top  \Phi,\; \Phi(0)=I_p\\
			\lab{thet}
			\dot{\hat \theta} & =\gamma PT \Delta [\caly -\Delta \calg(\hat\theta) ],\; \hat\theta(0)=\theta_0 \in \rea^q,
		}
		\endsubequ
		with tuning gains $\gamma_g>0$, $\gamma >0$,  and we defined
		\begsubequ
		\lab{aydelt}
		\begali{
			\lab{delt}
			\Delta & :=\det\{I_p-\Phi\}\\
			\lab{yt}
			\caly & := \adj\{I_p- \Phi\} [\hat\theta_g -   \Phi\theta_{g0}],
		}
		\endsubequ
		where $ \adj\{\cdot\}$ denotes the adjugate matrix. Then, for all $\theta_{g0} \in \rea^p$ and $\theta_{0} \in \rea^q$, we have the exponential convergence
		\begequ
		\lab{parcon}
		\lim_{t \to \infty}\tilthe(t)=0,
		\endequ
		with all signals \textit{ bounded}.
	\end{lemma}
	\subsection{Main adaptive stabilization result}
	\lab{subsec42}
	%
	In the proposition below we present an indirect adaptive controller for the system \eqref{sys} that ensures UAS of the closed-loop. 
	
	\begin{proposition}
		\label{pro3}
		Consider the nonlinearly parameterized, nonlinear system \eqref{sys} satisfying Assumptions  {\bf A1}-{\bf A5}, with $\Omega$, defined in \eqref{dynextome} of Proposition \ref{pro1}, satisfying Assumption  {\bf A7}.
		Let the adaptive control be given by
		$$
		u =\beta(x ,\hat \theta,t),
		$$
		where $\hat \theta$ is generated via  the G+D parameter estimator of Lemma \ref{lem1}. Then, the zero equilibrium of the adaptive error system \eqref{errsysada} is UAS. Consequently,    \eqref{asycon} holds with all signals bounded provided the initial errors $\tilde x(0)$ are sufficiently small.
	\end{proposition}
	
	\begin{proof}
		First, notice that using \eqref{lre} the error equation for the estimator is given by
		$$
		\dottilthe = -\gamma \Delta^2 PT[ \calg(\tilde\theta+\theta) - \calg(\theta)]=:F_2(\tilde \theta,t).
		$$
		In \cite[Poposition 7]{WANetal}the Lyapunov function candidate
		$$
		V(\tilde\theta):=\hal \tilde\theta^\top P^{-1}\tilde\theta,
		$$
		is used  to show that, under Assumptions {\bf A4} and  {\bf A7},  its origin is \textit{ globally exponentially stable} .
		
		Second, the state equation of the closed-loop system takes the form
		\begalis{
			\dot {\tilde x} & =f_\star(\tilde x,\theta,t)+\chi(\tilde x,\tilde \theta,t),
		}
		where we defined the perturbation term
		$$
		\chi(\tilde x,\tilde \theta,t) := g(\tilde x+x_\star,\theta)\begmat{ \beta(\tilde x+x_\star,\tilde \theta+\theta,t)  -\beta(\tilde x+x_\star,\theta,t)\\ 0},
		$$
		which satisfies $\chi(\tilde x,0,t)=0$. The overall dynamics of the closed-loop system clearly has a cascade form 
		\begali{
			\nonumber 
			\dot {\tilde x} &  =  F_1(\tilde x,\tilde \theta,t)\\
			\lab{cassys} 
			\dot{\tilde{\theta}}&= F_2(\tilde \theta,t),
		}
		with $F_1(0 ,0,t)=0$ and $F_2(0,t)=0$. Moreover, in Lemma \ref{lem1} it is shown that all signals of the G+D estimator are bounded, consequently there exists a constant $c>0$ such that
		$$
		\sup_{t \geq 0} \; \sup_{|\tilde \theta| \leq c}\|\nabla_{\tilde \theta}F_2(\tilde \theta,t)\| < \infty,
		$$
		where $\|\cdot\|$ is the induced matrix norm. Assumption {\bf A5} ensures that the origin of the $\tilde x$ subsystem is an UAS  equilibrium of the unperturbed system. Invoking \cite[Theorem 3.1]{VID} we conclude that the closed-loop system \eqref{cassys} has a UAS equilibrium at  the origin. 

	\end{proof}
	
	%
	\section{Examples}
	\label{sec5}
	%
	In this section the application of the indirect adaptive controller of Proposition \ref{pro3} is illustrated with two different examples.
	
		\subsection{A port-Hamiltonian system}
	\label{subsec51}
Consider the LTI  port-Hamiltonian system
	\begali{
		\nonumber
		\dot x &  =\begmat{0 & -a \\ a & 0}x+\begmat{\theta \\ \theta^2}u\\
		\lab{phsys}
		y_p&=\begmat{\theta & \theta^2}x,
	}
	where $a$ is known and the state is measurable, hence ensuring Assumption {\bf A1} of Proposition \ref{pro3}. The system is $u \mapsto y_p$ is passive with storage function $S(x)=\hal |x|^2$ and admits the NLPRE \eqref{linpars} with
	$$
	\phi_s=u x,\; \calg_s(\theta)= \begmat{\theta \\ \theta^2}.
	$$
	Hence Assumptions {\bf A2} and  {\bf A3} are satisifed.  Clearly, selecting $T=\begmat{1 & 0}$ and $P=1$ ensures the monotonicity Assumption {\bf A4} with $\rho=1$. 
	
	Assume the control objective is to stabilize the zero equilibrium. The closed-loop polynomial  for a linear state feedback of the form $u=-{1 \over \theta} k^\top x$ is given by
	$$
s^2 +(k_1+k_2\theta)s+ a(a-\theta k_1+k_2),	
	$$
which is a Hurwitz polynomial for $k=\col(1,\theta)$. Therefore, the static state feedback
	$$
	\beta(x,\theta) = - {1 \over \theta}  x_1 - x_2,
	$$
ensures  Assumption {\bf A5}. Finally, since the regressor of the NLPRE \eqref{lre} is given by
	$$
	\Omega = -F(\calp)[\beta(x,\hat \theta) x],
	$$
	it is clear that the interval excitation  Assumption {\bf A7} holds for all $x(0) \neq 0$.  Since all assumptions of Proposition \ref{pro3} are satisfied applying the adaptive controller
	$$
	u=\beta(x,\hat \theta),
	$$ 
	to the system  \eqref{phsys}, with $\hat \theta$ generated with the G+D estimator of Lemma \ref{lem1}, ensures UAS of the closed-loop system. 
	
	In contrast with the situation above if we adopt the standard linear parameterization of Subsection \ref{subsec32} and apply a gradient estimator to the overparameterized LRE the resulting adaptive controller will fail. Indeed,   for the system \eqref{phsys} the overparameterized LRE \eqref{lres} is satisfied with
	\begalis{
		{\bf Y}&:= \calp F(\calp)[x]- F(\calp)\Bigg[ \begmat{-a x_2 \\ a x_1}\Bigg] \\
		{\bm \Omega}  &:=F(\calp)[u]\\
	\Theta&:= \begmat{\theta \\ \theta^2}.
	}
	Therefore, the error dynamics for the gradient estimator have the form 
	\begin{subequations}
		\begin{align}
		\dot{\tilde\Theta}=&-\gamma \Omega^2 \tilde \Theta\label{error1}\\
		\dot x_{\Omega}=& -\lambda x_\Omega + u\label{error2}\\
		\Omega=&\lambda x_{\Omega}
	\end{align}
\end{subequations}
where the solution of \eqref{error1} is 
$$\tilde \Theta(t)= \tilde\Theta(0)e^{-\gamma\int_0^t \Omega^2(\tau)d\tau}. $$
	It follows that the gradient estimator of $\Theta$ will ensure parameter convergence if and only if $u(t) \not\in \call_2$. This condition will not be satisfied since, according to the control objective,  it is desired that $x\to 0$ which implies that $u\to 0$.
\subsection{An electrical circuit}
\label{subsec52}

\noindent Consider the electrical circuit depicted in Fig. \ref{fig1}. The dynamics of this system is described by
\begin{equation}\label{circuit}
\begin{aligned}
	\begin{bmatrix}
		\theta_1&0\\0&\theta_1^\alpha 
	\end{bmatrix}	\dot x=& 
	\begin{bmatrix} 
		0& 0\\0&-\theta_2 \end{bmatrix}x+\begin{bmatrix}-x_2\\x_1\end{bmatrix}u+\begin{bmatrix}E\\0\end{bmatrix}\\
	y_p=&x_1
\end{aligned}
\end{equation}
where the physical meaning of the state vector $x$ and the parameters $\theta_i >0$ and $E>0$ are given in the figure, with $\alpha\in\mathbb{R}$.

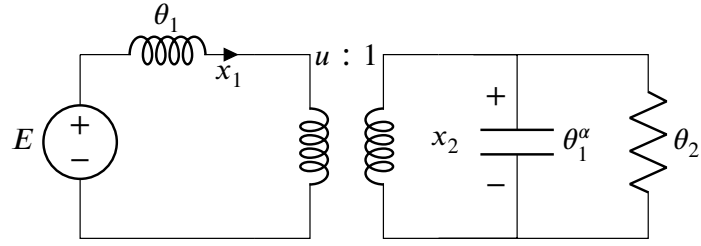
\begin{figure}
	\centering
	\scalebox{1.15}{
		\begin{circuitikz}[scale=1, american voltages, american currents]
			\draw
			(0,2) to[V, l_=\mbox{$E$}] (0,0)
			(0,2) to[cute inductor, l=$\theta_1$,i_=$x_1$] (2,2)
			(2,2) node[transformer, anchor=A1](T) {}
			(T.base) node {$u:1$}
			(T.A2) to[short] (T.A2 -| 0,0) to [short] (0,0)
			(T.B1) to [short] (5,2)
			to [C, l=$\theta_1^\alpha$, v_=$x_2$] (5,0)
			(T.B2) to[short] (T.B2 -| 5,0 ) to [short] (5,0)
			(4.5, 2) to[short] (6.5,2)
			to[resistor, l=$\theta_2$] (6.5,0)
			(T.B2) to[short] (T.B2 -| 6.5,0) to[short] (6.5,0);
	\end{circuitikz}}
	\caption{System \eqref{circuit}.}\label{fig1}
\end{figure}

The system $E \mapsto x_1$ is passive with storage function
$$
S(x)=\frac{1}{2}(\theta_1x_1^2+\theta_1^\alpha x_2^2).
$$
Moreover, it admits the NLPRE \eqref{linpars} with
\begin{align}
	b_s(x)=Ex_1,\;\phi_S(x)=\frac{1}{2}\begin{bmatrix}x_1^2\\ x_2^2\end{bmatrix},\;\phi_d(x)=x_2^2,\;\mathcal{G}_s(\theta)=\begin{bmatrix}\theta_1\\\theta_1^\alpha\\\end{bmatrix},\; \mathcal{G}_d(\theta)=\theta_2.
\end{align}

Hence Assumptions {\bf A2} and  {\bf A3} are satisifed.  Clearly, selecting 
$$
T=\begmat{1 & 0& 0\\0&0&1}.
$$ 
and $P=\rho I_2$ ensures the monotonicity Assumption {\bf A4} for any $\rho > 0$. 

The set of assignable equilibria is given by
\begequ
\lab{assequ}
\{x \in \rea^2\;|\; Ex_1-\theta_2 x^2_2=0\}.
\endequ
Assume the control objective is to regulate the voltage at $x_2=x_{2\star}:=\kappa\in\mathbb{R}$.  It is possible to show that with the static state feedback
\begin{align}\label{beta2}
	u= \beta(x,\theta) := -k_p\Big(\frac{\theta_2\kappa^2}{E}x_2-\kappa x_1 \Big) + \frac{E}{\kappa}
\end{align}
with free gain $k_p>0$, the task is accomplished---that is, Assumption {\bf A5} is satified. To prove it, notice first that from the assignable equilibrium set  \eqref{assequ},  we get the value of $x_1$ at the equilibrium as $x_{1\star}:=\frac{\theta_2}{E}\kappa^2$. Also, from the first equation  of \eqref{circuit} with $\dot{x}=0$, we obtain the value of $u$ at the equilibrium, that is,  $u_\star:=\frac{E}{\kappa}$. Now, with $\tilde{(\;)}:={(\;)-(\;)_\star}$, consider the Lyapunov function $W=S(\tilde x)$. Setting $\dot x=0$ in \eqref{circuit}, we get 
\begin{align}
	\begin{bmatrix}E\\0	\end{bmatrix}= -\begin{bmatrix} 0&0\\0&-\theta_2\end{bmatrix} x_\star - \begin{bmatrix} -x_{2\star}\\x_{1\star}\end{bmatrix}u_{\star},
\end{align}
which substituted into \eqref{circuit} produces
\begin{align}
	\begin{bmatrix}\theta_1&0\\0&\theta_1^\alpha	\end{bmatrix}\dot {\tilde x} =  \begin{bmatrix} 0&0\\0&-\theta_2\end{bmatrix} \tilde x + \begin{bmatrix} -x_{2\star}\\x_{1\star}\end{bmatrix}\tilde u + \begin{bmatrix} -\tilde {x}_{2}\\\tilde {x}_{1}\end{bmatrix}( \tilde u +u_\star ).
\end{align}
Thus, the time derivative of $W$ is
\begin{align*}
	\dot{W}=&-\theta_2 \tilde x_2^2 + \left(\frac{\theta_2\kappa^2}{E}\tilde{x}_2- \kappa \tilde{x}_1\right)\tilde u.
\end{align*}
\begin{figure*}[t!] 
	\centering
	\subfloat[\scriptsize{Estimated parameters.}\label{fig2}]{{\centering %
			\includegraphics[width=0.45\textwidth]{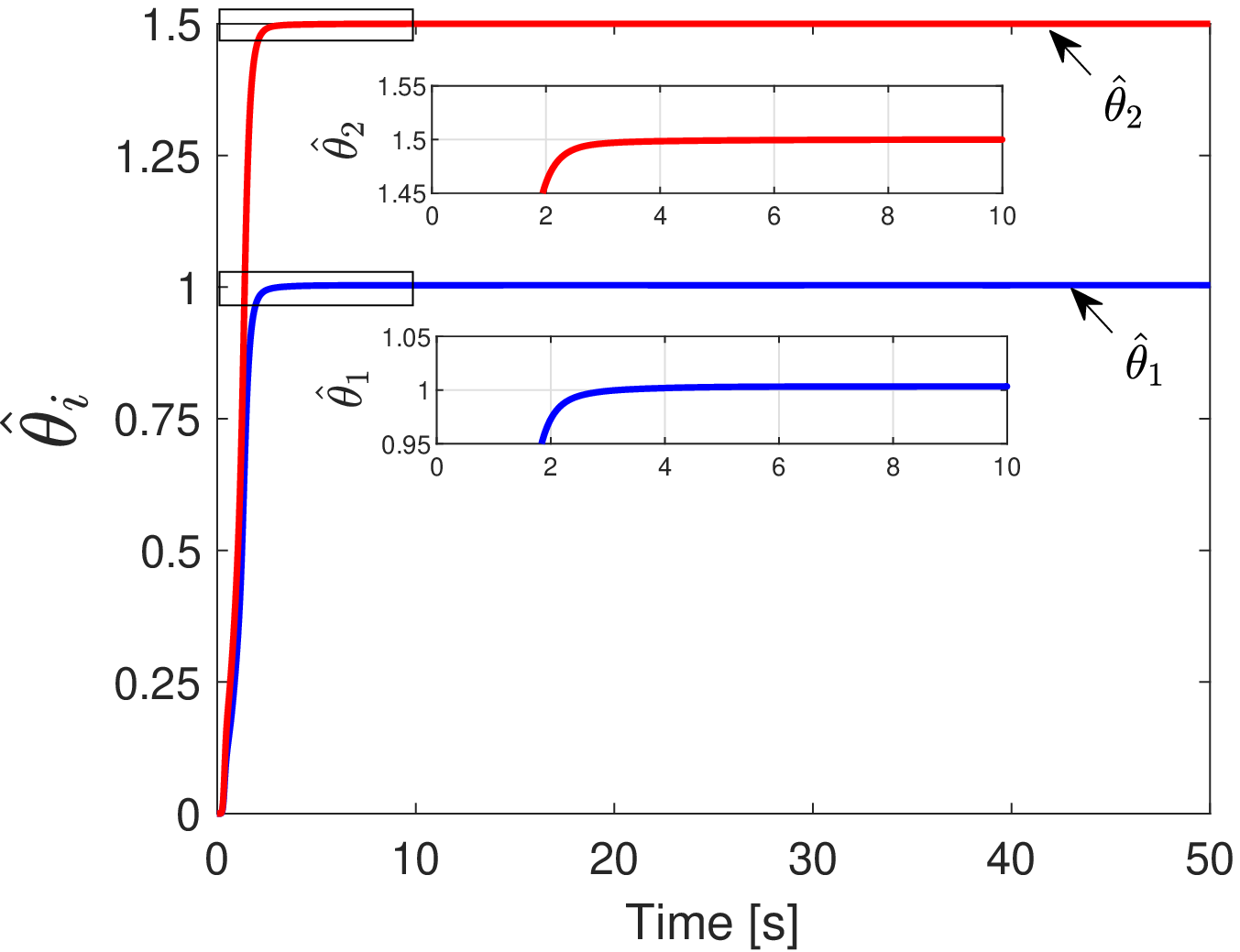}}}\hspace{1.5cm}
	\subfloat[\scriptsize{Closed-loop regulation performance.}\label{fig3}]{{\centering %
			\includegraphics[width=0.45\textwidth]{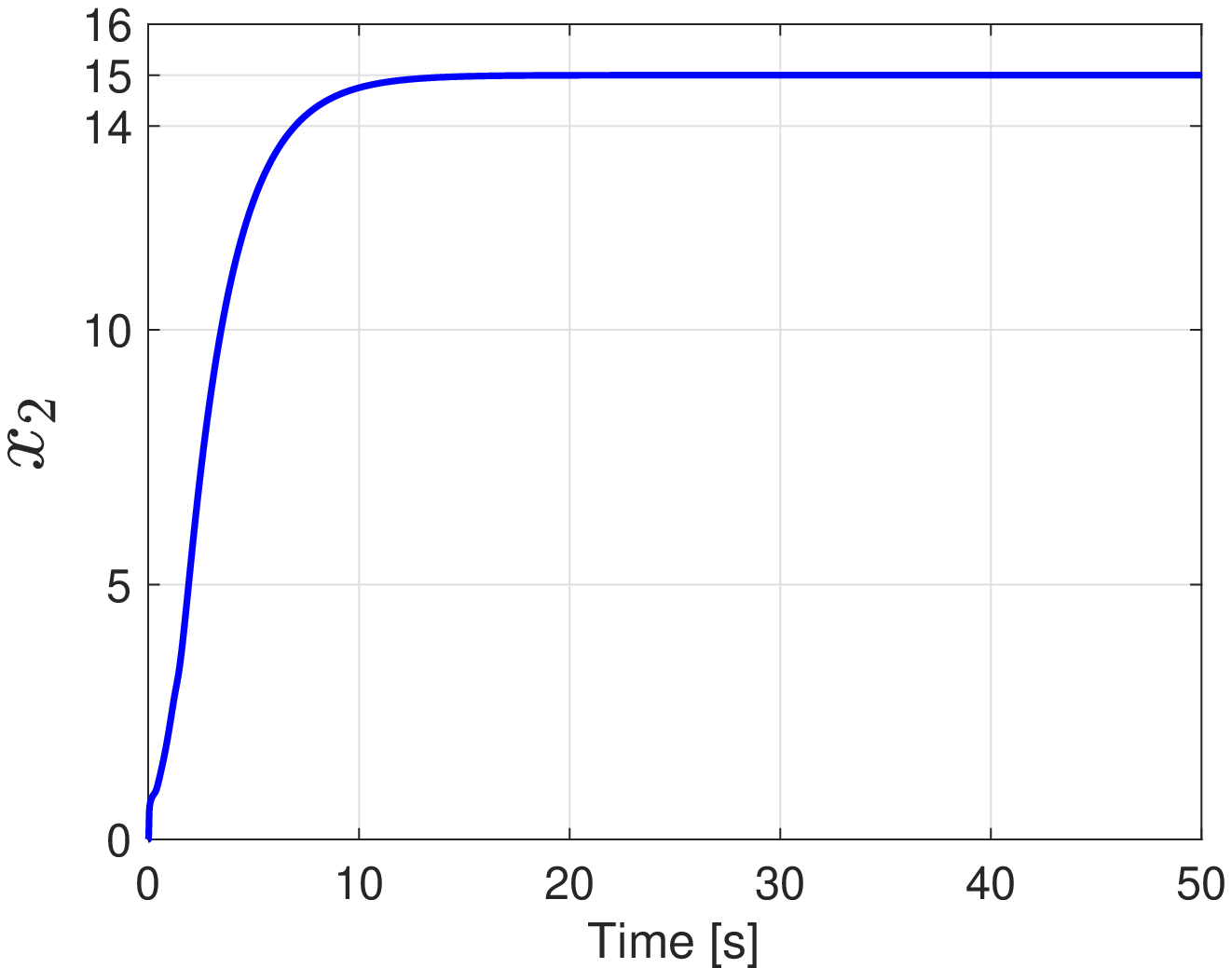}}}
	\caption{\scriptsize{Simulation results for power balance parametrization and G+D estimator.}}\label{sims1}
\end{figure*}
\begin{figure*}[t!] 
	\centering
	\subfloat[\scriptsize{Estimated parameters.}\label{fig4}]{{\centering %
			\includegraphics[width=0.45\textwidth]{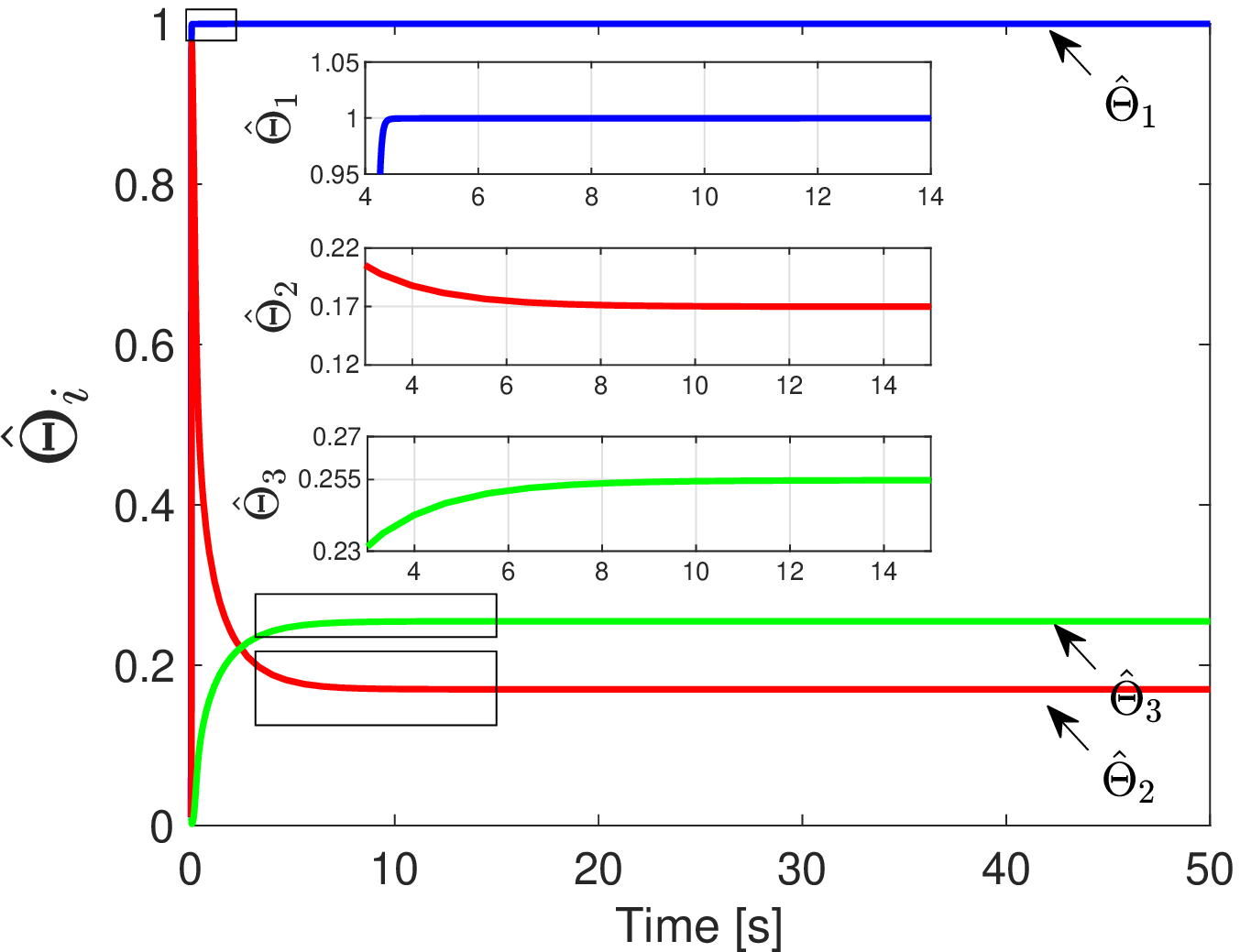}}}\hspace{1.5cm}
	\subfloat[\scriptsize{Closed-loop regulation performance.}\label{fig5}]{{\centering %
			\includegraphics[width=0.45\textwidth]{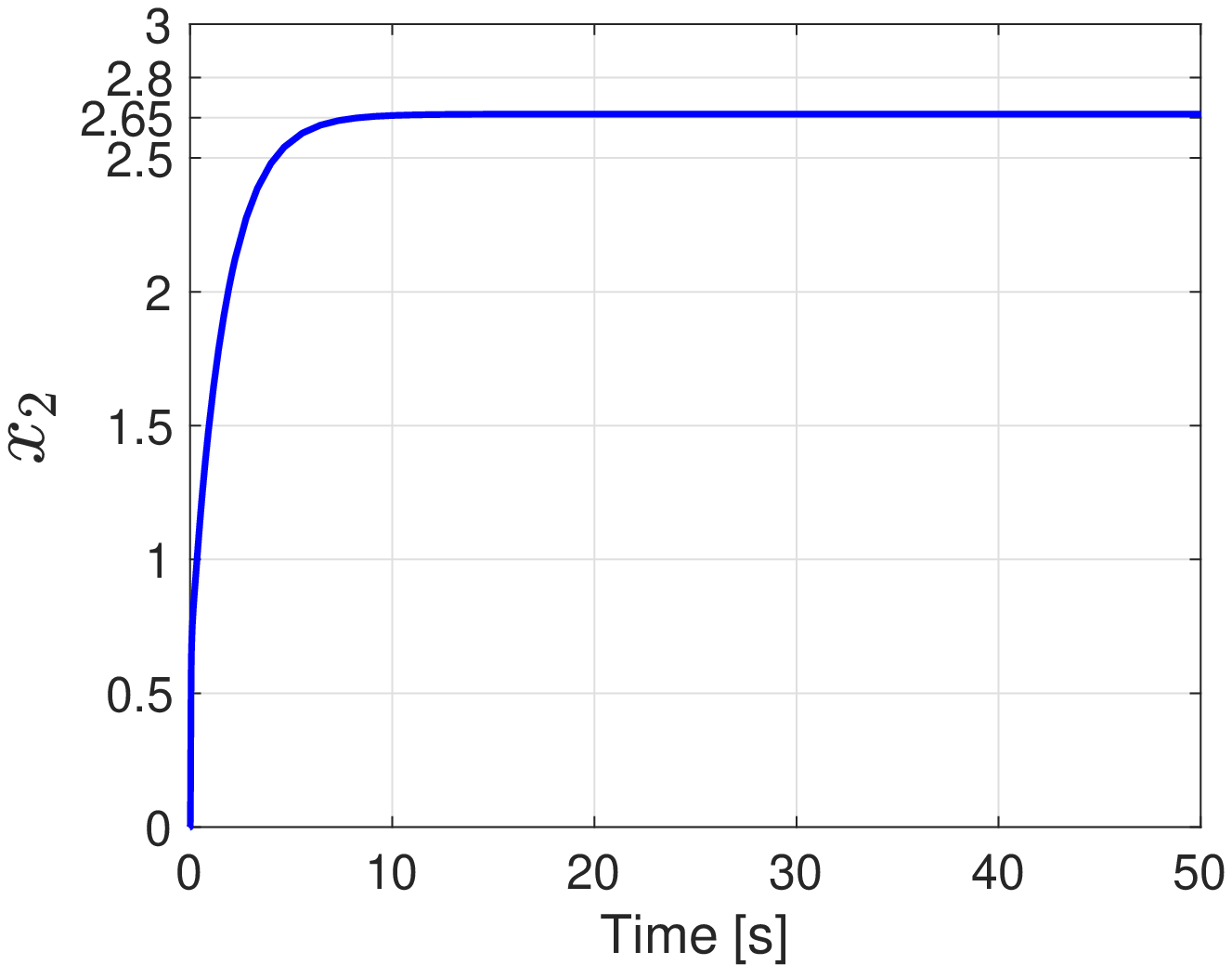}}}
	\caption{\scriptsize{Simulation results for gradient estimator  \eqref{graest} and standard parametrization.}}\label{sims2}
\end{figure*}

From \eqref{beta2} and the definition of $u_\star$, it follows that $\tilde u= \beta - u_\star$. Substituting the later into the last equation produces
\begin{align*}
	\dot{W}=&-\theta_2\tilde x_2^2 - k_p (\frac{\theta_2\kappa^2}{E}\tilde{x}_2- \kappa \tilde {x}_1) (\frac{\theta_2\kappa^2}{E}x_2- \kappa x_1)\\
	=&-\theta_2\tilde x_2^2 - k_p (\frac{\theta_2\kappa^2}{E}\tilde{x}_2- \kappa \tilde {x}_1)^2,
\end{align*}
where, invoking LaSalle's Invariance principle, we can conclude that $x\to x_\star$.

To implement the adaptive controller of Proposition \ref{pro3} we compute from \eqref{dynext}
$$
Y= E F(\calp)[x_1],\;\Omega = \begin{bmatrix}\frac{1}{2} pF(\calp)\Bigg[\begin{bmatrix}x_1^2\\ x_2^2\end{bmatrix}\Bigg]\\F(\calp)[x_2^2]\end{bmatrix},
$$
and define the mapping $\mathcal{G}(\theta)=\col(\theta_1, \theta_1^\alpha, \theta_2)$.

On the other hand, the standard parameterization of Proposition \ref{pro2} is computed with
\begalis{
	{\bf Y}&= \calp F(\calp)[x] \\
	{\bm \Omega}^\top  &=F(\calp)\Bigg[\begmat{-x_2 u + E &0 &0\\ 0 &x_1u&-x_2 }\Bigg],
}
with the overparameterized vector $\Theta= \col\Big({1 \over \theta_1},{1 \over \theta_1^\alpha},{\theta_2 \over \theta_1^\alpha}\Big)$.

For the simulations we consider the adaptive controller $u=\beta(x,\hat \theta)$,  where the estimate $\hat\theta$ is generated either by the G+D estimator of Lemma \ref{lem1} or from a standard gradient estimator for the overparameterized LRE \eqref{lres}. That is,
\begequ
\lab{graest}
\dot {\hat \Theta}=\gamma	{\bm \Omega} (	{\bf Y} - {\bm \Omega}^\top \hat \Theta).
\endequ
The parameter values of the system are $\theta_1=1,\theta_2=1.5, E=15, \alpha=2$, $k_p=10$ and $\kappa=15$.

The gains of the G+D estimator and filter constant were taken as $\gamma_g=100$, $\gamma=50$ and $\lambda=10$, respectively. The results of the simulation are shown in  Fig. \ref{sims1}. As seen from the figures the parameter estimates converge to their respective (Fig. \ref{fig2}) values and  $x_2(t) \to\kappa=15$ (Fig. \ref{fig3}), as desired.

In contrast with the situation above if we adopt the standard linear parameterization of Subsection \ref{subsec32} and apply the gradient estimator \eqref{graest} to the overparameterized LRE the resulting adaptive controller will fail. This comes from the fact that the regulation task requires that $x\to x_\star$. Therefore, from \ref{beta2}, $u\to u_\star=\frac{E}{\kappa}$. Since $x$ and $u$ converge to a constant, then $\Omega\to \Omega_\star$---that is, $\Omega$ converges to a constant as well. Thus, $\Omega$ is not PE.

The simulation results of this second scenario are shown in Fig. \ref{sims2} for estimator gain $\gamma=30$ and the remaining parameters selected as before. In Fig. \ref{fig4}, the estimation perfomance is shown for $\Theta$. Zooming in the plot of the estimated parameters in this figure, it can be seen that $\hat \Theta(t) \to \col(1.00, 0.17, 0.25)$, however, $\Theta=\col(1.00, 1.00, 1.50)$. That is, the estimation is deficient since $\hat{\Theta}$. This lead to an erroneous estimate of $\theta_2$ and, as consequence of that, the regulation performance is poor. This is evident in Fig. \ref{fig5} where $x_2$ is not driven to its setpoint $\kappa=15$.

	\section{Conclusions}
	\label{sec6}
	%
	We have presented in the paper a procedure to identify the parameters of nonlinear, nonlinearly parameterized, dissipative systems of the form \eqref{sys}. The method is based on the power balance equation of the system \eqref{powbal}, avoiding in this way the messy computations and stringent excitation requirements related with the standard parameterization of the systems vector field and input matrix given in \eqref{linparsys}. Invoking the G+D parameter estimator proposed in \cite{WANetal}, which ensures global exponential convergence of the parameter errorunder very weak regressor excitation assumptions, we proposed an indirect adaptive controller that guarantees UAS of the closed-loop system.
	
	Adaptive control was one of the main research topics in control  from the 70s to the mid-90s. The development and analysis of the problem's many and varied solutions over all these years, unquestionably played an absolutely major role in guiding us to our present understanding of control theory in general.  Although many critical issues remained open, a large part of the control community moved away from the field. Partly responsible for this unfortunate situation was the deviation of the problem formulation from a self-tuning procedure to the---more mathematically tractable but of  little practical relevance---\textit{ stabilization technique}. 
	
	For at least two reasons it is reasonable to expect a renewed interest in adaptive control and identification theories in the near future.  On one hand, to comply with the increasing performance requirements imposed to modern control systems it is necessary to develop efficient controller tuning procedures, that a well-formulated, adaptive control theory can provide. On the other hand, in recent years we have witnessed an explosion of references to the hyped-up \textit{ artificial intelligence} field, which is simply  the application of a neural network-based structure to a massive collection of data, whose success \textit{ in some particular applications} has been widely publicized. Obviously, the interest for a scientific theoretical field of ``procedures that work in some examples" is highly questionable. Because of the prevalence of nonlinear parameterizations in neural networks, very little theoretical understanding is available on adaptive neural networks---a situation that was already denounced 25 years ago \cite{ORT} and is still prevalent. Development of solid theoretical foundations are essential to turn this tide, an endeavour where adaptive control should play a central role.


\bibliography{bibliography}

\begin{thebibliography}{10}
\providecommand \doibase [0]{http://dx.doi.org/}%

\bibitem{ASTKARORT}
Astolfi A, Karagiannis D, Ortega R. \textit{Nonlinear and Adaptive Control
  Design with Applications}. In: London: Springer-Verlag.  2007.

\bibitem{KRSKANKOKbook}
Krstic M, Kanellakopoulos I, Kokotovic P. \textit{Nonlinear and Adaptive
  Control Design}. In: New York: John Wiley \&Sons.  1995.

\bibitem{MARTOMbook}
Marino R, Tomei P. \textit{System Identification: Theory for the User}. In:
  Upper Saddle River: Prentice Hall.  1995.

\bibitem{PRAetal}
Praly L, Bastin G, Pomet JB, Jiang Z. \textit{Adaptive stabilization of
  nonlinear systems}, Foundations of Adaptive Control. In:  1991.

\bibitem{LJUbook}
Ljung L. \textit{System Identification: Theory for the User}. In: New Jersey:
  Prentice Hall.  1987.

\bibitem{SASBODbook}
Sastry S, Bodson M. \textit{Adaptive Control: Stability, Convergence and
  Robustness}. In: New Jersey: Prentice Hall.  1989.

\bibitem{ORTetal}
Ortega R, Gromov E, Nu\~no E, Pyrkin A, Romero J. Parameter estimation of
  nonlinearly parameterized regressions: Application to adaptive control. {\it
  Automatica} 2021\string; 127.

\bibitem{ORT}
Ortega R. Some remarks on adaptive neuro-fuzzy systems. {\it Internatonal
  Journal of Adaptive Control and Signal Processing} 1996\string; 10(2)\string:
  79-83.

\bibitem{ARAetalijc}
Aranovskiy S, Ortega R, Romero J, Sokolov D. A globally exponentially stable
  speed observer for a class of mechanical systems: experimental and simulation
  comparison with high-gain and sliding mode designs. {\it International
  Journal of Control} 2019\string; 92(7)\string: 1620-1633.

\bibitem{SLOLIaut}
Slotine J, Li W. Composite adaptive control of robot manipulators. {\it
  Automatica} 1989\string; 25(4)\string: 509-519.

\bibitem{KHADOMbook}
Khalil W, Dombre E. \textit{Modeling, Identification \& Control of Robots}. In:
  London: Butterworth-Heinemann.  2004.

\bibitem{ROMORTBOB}
Romero J, Ortega R, Bobtsov A. Parameter Estimation and adaptive control of
  Euler-Lagrange systems using the power balance equation parameterization.
  {\it International Journal of Control} 2021.

\bibitem{BOBetal}
Bobtsov A, Yi B, Ortega R, Astolfi A. Generation of new exciting regressors for
  consistent on-line estimation of a scalar parameter, (submitted). {\it IEEE
  Transactions on Automatic Control} 2021\string; ({\tt arXiv:2104.02210}).

\bibitem{BOFSLO}
Boffi NM, Slotine JJE. Higher-order algorithms and implicit regularization for
  nonlinearly parameterized adaptive control. {\it MIT Int. Report}
  2020\string; {(\tt arXiv:1912.13154v3}).

\bibitem{WANetal}
Wang L, Ortega R, Bobtsov A, Romero J, Yi B. Identifiability implies robust,
  globally exponentially convergent on-line parameter estimation: Application
  to model reference adaptive control (submitted). {\it IEEE Transactions on
  Automatic Control} 2021\string; {\tt{(arXiv:2108.08436)}}.

\bibitem{KRERIE}
Kreisselmeier G, Rietze-Augst G. Richness and excitation on an interval---with
  application to continuous-time adaptive control. {\it IEEE Transactions on
  Automatic Control} 1990\string; 35(2)\string: 165-171.

\bibitem{TAObook}
Tao G. \textit{Adaptive {C}ontrol {D}esign and {A}nalysis}. In: New Jersey:
  John {W}iley \& {S}ons.  2003.

\bibitem{MOYbook}
Moylan P. \textit{Dissipative {S}ystems and {S}tability}. In: Springer.  2014.

\bibitem{VANbook}
Schaft v.~dA. \textit{$L_2$-Gain and Passivity Thechniques in Nonlinear
  Control}. In: Springer International Publishing.  2016.

\bibitem{DEM}
Demidovich BP. Dissipativity of nonlinear systems of differential equations (In
  Russian). {\it Vestnik Moscow State University, Ser. Mat. Mekh., Parts: I-6
  and II-1} 1961\string: P.I-6:19-27, P.II-1:3-8.

\bibitem{PAVetal}
Pavlov A, Pogromsky A, Wouw v.~dN, Nijmeijer H. Convergence dynamics, a tribute
  to Boris Pavlovich 
  2004\string; 52(3)\string: 257-261.

\bibitem{VANJELbook}
Schaft v.~dA, Jeltsema D. \textit{Port-Hamiltonian Systems Theory: An
  Introductory Overview}. In: Foundations and Trends in Systems and Control.
  New Jersey: Now Publishers.  2014.

\bibitem{VID}
Vidyasagar M. Decomposition techniques for large-scale systems with
  non-additive interactions: stability and stabilizability. {\it IEEE
  Transactions on Automatic Control} 1980\string; 25(4)\string: 773-779.

\end{thebibliography}
%

\end{document}